\documentclass[journal,onecolumn]{IEEEtran}
\IEEEoverridecommandlockouts
\usepackage{cite}
\makeatletter
\newcommand{\removelatexerror}{\let\@latex@error\@gobble}
\makeatother
\usepackage{amsmath,amssymb,amsfonts}
\usepackage{subfigure}
\usepackage{textcomp}

\usepackage{bm}
\usepackage{diagbox}
\usepackage{booktabs}
\usepackage{multirow}
\usepackage{geometry}
\newgeometry{
  top=72pt,
  bottom = 54pt,
  left = 54pt,
  right = 54pt
}
\usepackage{array}
\usepackage[caption=false,font=normalsize,labelfont=sf,textfont=sf]{subfig}
\usepackage{textcomp}
\usepackage{stfloats}
\usepackage{url}
\usepackage{diagbox}
\usepackage{verbatim}
\usepackage{graphicx}
\usepackage{multirow}
\usepackage{cite}
\usepackage{xcolor}
\usepackage{amsthm}

\newtheorem{theoremL}{Lemma}
\newtheorem{theoremC}{Assumption}

\newtheorem{theoremT}{Theorem}
\newtheorem{theoremD}{Definition}
\newtheorem{lemma}[theoremL]{Lemma}
\newtheorem{assumption}[theoremC]{Assumption}

\newtheorem{theorem}[theoremT]{Theorem}
\newtheorem{Definition}[theoremD]{Definition}
\usepackage{stfloats}
\usepackage[lined,boxed,commentsnumbered,ruled]{algorithm2e}
\def\BibTeX{{\rm B\kern-.05em{\sc i\kern-.025em b}\kern-.08em
    T\kern-.1667em\lower.7ex\hbox{E}\kern-.125emX}}

\begin{document}

\title{Is Locational Marginal Price All You Need for Locational Marginal Emission?}

\author{\IEEEauthorblockN{Xuan He\IEEEauthorrefmark{1},
Danny H.K. Tsang\IEEEauthorrefmark{1}, and Yize Chen\IEEEauthorrefmark{2}\thanks{X.He and D.H.K. Tsang are with the Hong Kong University of Science and Technology (Guangzhou), China. Y. Chen is with the University of Alberta, China. Emails: xhe085@connect.hkust-gz.edu.cn, eetsang@ust.hk,  yize.chen@ualberta.ca}}}

\maketitle

\begin{abstract}
Growing concerns over climate change call for improved techniques for estimating and quantifying the greenhouse gas emissions associated with electricity generation and transmission. Among the emission metrics designated for power grids, locational marginal emission (LME) can provide system operators and electricity market participants with valuable information on the emissions associated with electricity usage at various locations in the power network. In this paper, by investigating the operating patterns and physical interpretations of marginal emissions and costs in the security-constrained economic dispatch (SCED) problem, we identify and draw the exact connection between locational marginal price (LMP) and LME. Such interpretation helps instantly derive LME given nodal demand vectors or LMP, and also reveals the interplay between network congestion and nodal emission pattern. Our proposed approach helps reduce the computation time of LME by an order of magnitude compared to analytical approaches, while it can also serve as a plug-and-play module accompanied by an off-the-shelf market clearing and LMP calculation process.
\end{abstract}


\section{Introduction}
Electricity sector is the largest greenhouse gas (GHG) emission producer, contributing $29\%$ of worldwide emissions annually~\cite{IEA2023}.
The pressing climate change goals require a swift low-carbon transition for power and energy systems. So far, an increasing number of ongoing and future power system projects are taking carbon emission reduction as one of the top priorities \cite{liu2023monitoring,masanet2020recalibrating}. To quantify the impact of carbon reduction operation and planning projects, accurate emission measurement and estimation techniques are necessary for both policymakers \cite{mallapaty2020china, perissi2022investigating, xu2024united} and industries \cite{Google, anthony2020carbontracker}. 

Though there has been a rich literature on quantifying the macro-level or whole network's operational emissions~\cite{li2024review, lamb2021review} and system-level marginal emissions~\cite{holland2022marginal, wang2016estimating}, the focus on localized emission assessments remains limited. The calculation of locational marginal emissions (LME) addresses this gap by quantifying the change in emission by nodal level demand change. An accurate estimation of LME is crucial for monitoring emission statuses, implementing demand-side sustainable management, and formulating effective carbon-reduction policies under network reliability considerations. Currently, PJM and some other system operators have published pioneering records of LME on an hourly basis~\cite{sofia2024carbon, PJM}, while requiring exact knowledge of transmission networks.
\begin{figure}[t]
	\centering
\includegraphics[width=0.8\linewidth]{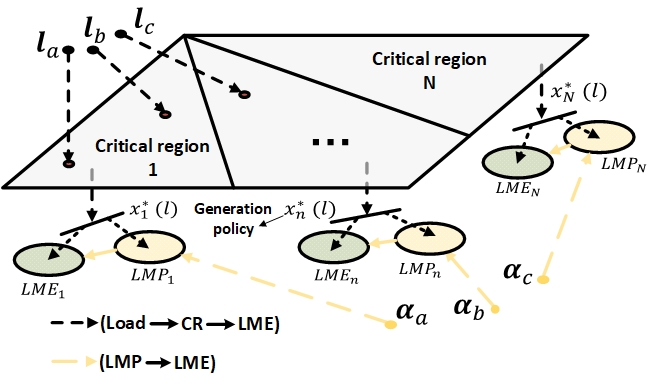}

\caption{\footnotesize Illustration of the mappings of Load-LME and LMP-LME via critical region projection, which helps conveniently derive LME for given SCED instances and load vectors.} \label{Framework}

\end{figure}
In the power system community, there have been rich and fruitful works on analyzing and utilizing the fundamental relationships between the optimal solution of the optimal power flow model (OPF) for the economic dispatch (ED) and locational marginal price (LMP)~\cite{kirschen2018fundamentals}. Given the load inputs, LMP reflects the total cost change associated with the increments of nodal load. Further, LMP can be derived as the dual variables while also reflecting generators' and lines' flow patterns. 
This motivates us to analyze LME based on OPF constraints and operating conditions given a load vector.

In this paper, the load-emission relationship is investigated explicitly based on the underlying security-constrained economic dispatch (SCED) model. More interestingly, by utilizing multi-parametric programming (MPP), we can conveniently partition the load spaces into critical regions~\cite{ji2016probabilistic}, while each region is attached to a specific LMP parameterization and a unique congestion pattern. Such characteristics have been invesitaged in electricity market~\cite{radovanovic2019holistic}, solving OPF~\cite{chen2022learning} and unit commitment problems~\cite{he2024efficient}. Moreover, the parametrization of critical regions also helps get the analytical solution of LME, circumventing computationally expensive \emph{finite difference} or \emph{implicit function theorem} implemented by the literature of LME calculation~\cite{valenzuela2023dynamic, ruiz2010analysis}, achieving speedups ranging from 22x to several thousand times. 

To the best of the authors' knowledge, this work is one of the first to draw the explicit mapping between LMP and LME. We give the condition under which the unique mapping between LMP and LME exist and can be derived. In real-world market environments, system operators often release real-time or day-ahead LMP, which can be readily used for deriving LME using our approach. Such a connection helps both system operators and energy end users to better understand the connections between network operation conditions (e.g., congestion and load change) with nodal marginal emission rates. This LMP to LME mapping is providing an efficient route for multiple entities with acceleration by an order of magnitude. For instance, for independent system operators (ISO) and regional transmission operators (RTO), it provides a plug-in tool for analyzing and revealing real-time emission information. For market participants and generation companies, our method provides an efficient way to find LME signals once LMP is exposed. In addition, our accurate LME estimation helps inform energy end users about their emission patterns associated with electricity usage.

\section{Problem Formulation}
Here we consider an ex-ante LMP model derived from security-constrained economic dispatch. Our model is based on the stylized approach in \cite{ji2016probabilistic}, which captures the process by which independent system operators compute LMP. Then, we can derive LMP and LME via sensitivity analysis.
\subsection{Security-Constrained Economic Dispatch Model}
Specifically, the system operator solves a DC-OPF problem to determine the optimal economic generation adjustments that meet forecasted loads and stochastic generation for the upcoming dispatch interval while satisfying generation and transmission constraints. We assume a connected power network with $n$ buses, $g$ generators, and $m$ lines. Let $\boldsymbol{x} \in \mathbb{R}^g$ be the generation and $\boldsymbol{l} \in \mathbb{R}^n$ be the nodal demand. We follow the typical modeling assumption to model power flows as a DC approximation, where $\mathbf{P}$ denotes the Power Transfer Distribution Factors (PTDF) matrix~\cite{wang2012computational}. $\mathbf{B}$ is a $n$ by $ g$ matrix mapping the generation to each bus. $\overline{\mathbf{f}}$ and $\underline{\mathbf{f}}$ denote the upper and lower bounds of line flow. Then, the SCED problem can be formulated as, 

\begin{subequations}
\label{DCOPF}
    \begin{align}
\min _{\boldsymbol{x}}\quad & \sum_{i=1}^g \;  c_i x_i \label{UC:obj_0}\\
\text { s.t. } \quad  &0 \leq \mathbf{x} \leq \mathbf{\bar{x}}, \label{generation_bound} \\
 &-\overline{\mathbf{f}} \leq \mathbf{P}(\mathbf{B}\mathbf{x}-\boldsymbol{l}) \leq \overline{\mathbf{f}}, \\
 &\sum_{i=1}^{g}x_{i} - \sum_{j=1}^{n}\boldsymbol{\ell}_j =0. \label{UC:balance_0}
\end{align}
\end{subequations}
where $C=\sum_{i=1}^g \;  c_i x_i $ denotes the total generation costs. Note our proposed method can also be applied to generators with quadratic or piecewise-affine costs.

\subsection{LMP and LME via Sensitivity Analysis}
Since LMP is interpreted as the cost of optimally supplying an increment (or decrement) of load at a particular location while satisfying all operational constraints. In the literature, LMP $\alpha_i$ can be calculated by taking the cost derivative
\begin{align} \label{LMP_ori}
\alpha_j = \frac{\partial C}{\partial l_j} = \frac{\partial(\sum_ic_ix_i)}{\partial l_j}=\sum_i(\frac{\partial x_i}{\partial l_j} \cdot c_i);
\end{align}

Similarly, once the optimal solution for the SCED is found, herein, we assume each generator has a fixed, known emission rate $e$ with units of $kg CO_2/MW$ associated with electric power generation. With the total carbon emissions rate $E$ given by $E = \sum_ie_ix_i$. When we treat the demand as an input variable, the output of total emission's sensitivity with respect to demand, namely the LME $\beta_i$, is calculated as
\begin{align} \label{LME_ori}
\beta_j = \frac{\partial E}{\partial l_j} = \frac{\partial(\sum_ie_ix_i)}{\partial l_j}=\sum_i(\frac{\partial x_i}{\partial l_j} \cdot e_i).
\end{align}
Thus, if we can get the analytical representation of $\frac{\partial x_i}{\partial l_j}$, we can directly calculate the LMP with \eqref{LMP_ori} and LME with \eqref{LME_ori}.

\section{Critical Region Projection-Based LMEs}
In this section, we develop the LME policies based on critical region projection (CRP) to project the load region to LMP and LME, which can circumvent computationally expensive calculations for individual load samples. This approach relies on multi-parametric programming (MPP) applied to the SCED model, where loads are treated as uncertain parameters and the load spaces can be partitioned into critical regions (CR). Each region will correspond to a specific pair of LMP and LME. First, we present the compact form of \eqref{DCOPF} to formulate its corresponding MPP model as follows,
\begin{subequations}\label{compact_DCOPF}
\begin{align}  
\min _{\boldsymbol{x}}\quad &\mathbf{c}^T\boldsymbol{x}\\
\mathbf{A}\mathbf{x}&  \leq \mathbf{b}  + \mathbf{F}\boldsymbol{l} \label{LME_MPP:Inequality}\\
\boldsymbol{x} &\in \mathcal{X} \subseteq \mathbb{R}^g \\
\boldsymbol{l} &\in \mathcal{L} \subseteq \mathbb{R}^n 
\end{align}
\end{subequations}
where $\mathbf{A}$ is a ($p \times g$) coefficient matrix for generation variables $\mathbf{x}$ and $\mathbf{F}$ is a ($p \times n$) coefficient matrix for load parameters $\boldsymbol{l}$. $\mathbf{b}$ is a constant vector of dimension $p$. The solution space $\mathcal{X}$ is defined by the set of constraints \eqref{generation_bound}-\eqref{UC:balance_0}, while here we let $\mathcal{L}$ be a predefined polyhedron $(1-\omega)\boldsymbol{l}_0\leq\boldsymbol{l} \leq(1+\omega)\boldsymbol{l}_0$. $\boldsymbol{l}_0$ is the nominal load profile and $\omega$ is a constant representing the operating variation range. By applying MPP, the critical regions within $\mathcal{L}$ can be identified. Within each of these regions, an affine policy maps the vector of load parameters to the corresponding solution for the generation variables.

\subsection{Affine Policy in Critical Load Regions}
Denote the Lagrange multiplier corresponding to each inequality in \eqref{LME_MPP:Inequality} as $\boldsymbol{\lambda} \in \mathbb{R}^p$, according to \cite{pistikopoulos2007multi}, we can utilize the following lemma to describe the relationship between $\mathbf{x}$ and $\boldsymbol{
l
}$ within the critical load regions defined later,
\begin{lemma} \label{lemma_1}
For a feasible $\boldsymbol{l}_\mathbf{c} \in \mathcal{L}$, in the neighborhood of the KKT point [$\mathbf{x}_\mathbf{c}, \mathbf{\lambda}_\mathbf{c}$], a first-order approximation of decision variable $\mathbf{x}$ and the Lagrange multiplier $\boldsymbol{\lambda}$ is,
\begin{align}
\left[\begin{matrix}
\mathbf{x}_\mathbf{c}(\boldsymbol{l})\\
\boldsymbol{\lambda}_\mathbf{c}(\boldsymbol{l})
\end{matrix}\right] = -(\mathbf{M}_\mathbf{c})^{-1}\mathbf{N}_\mathbf{c}(\boldsymbol{l}-\boldsymbol{l}_c)+ \left[\begin{matrix}
\mathbf{x}(\boldsymbol{l}_c)\\
\boldsymbol{\lambda}(\boldsymbol{l}_c)
\end{matrix}\right],
\end{align}
where
\begin{subequations} \label{Gradient_policy}
\begin{align}  
\mathbf{M}_\mathbf{c} &=
\left[
\begin{matrix}
\mathbf{0} & A^{T}_{1} & \cdots &  A^{T}_{p} \\
-\lambda_1A_1 & -V_1 & \cdots & 0 \\
\vdots & \vdots & \vdots & \vdots \\
-\lambda_pA_p & 0 & \cdots & -V_p \\
\end{matrix}
\right], \\
\mathbf{N}_\mathbf{c} &= [Y, \lambda_1F_1,...,\lambda_pF_p],\\
V_i &= A_i\boldsymbol{x}(\boldsymbol{l}_\mathbf{c})-b_i-F_i\boldsymbol{l}_\mathbf{c}.
\end{align}
\end{subequations}
and $Y$ is a null matrix of dimension ($g\times n$).
\end{lemma}

\begin{Definition}
The space of $\boldsymbol{l}$ where \eqref{Gradient_policy} holds is defined as a critical region, $CR^\mathbf{c}$, and can be formulated as follows,
\begin{subequations} \label{critical_region}
\begin{align}
CR^R &= \{\Tilde{A}\boldsymbol{x}_\mathbf{c}(\boldsymbol{l})\leq \Tilde{\boldsymbol{b}}+\Tilde{F}\boldsymbol{l}, \Tilde{\boldsymbol{\lambda}}_\mathbf{c}(\boldsymbol{l})\geq0, CR^{IG}\};\\
CR^{\mathbf{c}} &= \nabla{CR^R}.
\end{align}
\end{subequations}
where $\mathbf{\Tilde{A}}, \mathbf{\Tilde{b}}$, and $\mathbf{\Tilde{F}}$ correspond to inactive inequalities, and $CR^{IG}$ denotes a set of linear inequalities defining an initial given region. $\nabla$ is an operator which removes the redundant constraints.
\end{Definition}
Once $CR^{\mathbf{c}}$ has been defined for a solution $[\boldsymbol{x}(\boldsymbol{l_\mathbf{c}}), \boldsymbol{l_\mathbf{c}}]$, we continue to explore the rest of the region, $CR^{rest}$:
\begin{align}
CR^{rest} = CR^{IG}-CR^{\mathbf{c}}.
\end{align}
Then, another set of gradient policies in each of these regions and corresponding CRs are obtained. This procedure terminates when there are no more regions to explore. For the load vectors $\boldsymbol{l}$ from the same CR, they share the following sensitivity between load, cost, and generation:
\begin{align} \label{gradient_policy}
&\left(
\begin{matrix}
\frac{d\boldsymbol{x}(\boldsymbol{l})}{d\boldsymbol{l}} \\
\frac{d\boldsymbol{\lambda}(\boldsymbol{l})}{d\boldsymbol{l}}
\end{matrix}
\right) = -(\mathbf{M}_\mathbf{c})^{-1}\mathbf{N}_\mathbf{c}.
\end{align}

\subsection{Mapping from Load to LME and LMP}
In the electricity market literature, 
the LMP value at each bus is shown to be represented by the sum of the marginal price of generation at the reference bus and the marginal congestion price at the location associated with the active transmission constraints. For each critical region $r$, we formulate a new ($g \times n$) matrix $\mathbf{G}_{r}$ that involves the partial elements associated with the generation variables in \eqref{gradient_policy}. Then, the corresponding LMP can be given as,
\begin{align}  \label{LMP_Gr}
\boldsymbol{\alpha}_r = \mathbf{c}^{T}\mathbf{G}_r.
\end{align}
Similarly, the LME for a critical region can be given as follows,
\begin{align}  
\boldsymbol{\beta}_r = \mathbf{e}^{T}\mathbf{G}_r.
\end{align}
Then, we can project a feasible load vector $\Tilde{\boldsymbol{l}} \in \mathcal{L}$ to the corresponding LMP and LME of its critical region. Specifically, we identify the critical region $\Tilde{r}$ it belongs to using \eqref{critical_region}, which can be given as a mapping denoted by,
\begin{align} \label{LME_Gr}
    \Tilde{r} = CR(\Tilde{\boldsymbol{l}}).
\end{align}
With the index $\Tilde{r}$, we can match the LMP $\boldsymbol{\alpha}_{\Tilde{r}}$ and the LME $\boldsymbol{\beta}_{\Tilde{r}}$ for the load vector $\Tilde{\boldsymbol{l}}$.

\subsection{Mapping from LMP to LME}
By pre-recording the LMP and LME for each critical region, we can accurately determine the LME based solely on the released LMP vector, even without knowing the specific load vector or the corresponding critical region. This procedure's performance can be guaranteed by Theorem \ref{LMP_LME}.
\begin{assumption} \label{Assumption_1}
There will not be two matrices $\mathbf{G}_u$ and $\mathbf{G}_v$, which satisfies $\mathbf{c}^T\mathbf{G}_u=\mathbf{c}^T\mathbf{G}_v$ along with $\mathbf{e}^T\mathbf{G}_u \neq \mathbf{e}^T\mathbf{G}_v$. 
\end{assumption}
\begin{theorem} \label{LMP_LME}
Based on Assumption \ref{Assumption_1}, given the MPP model \eqref{compact_DCOPF} representation of SCED, there exists a unique LME for each LMP, that is
\begin{align} \label{mapping_LMP_LME}
    \boldsymbol{\beta} = \Phi(\boldsymbol{\alpha}).
\end{align}
\end{theorem}
\begin{proof} Assume that one LMP $\Tilde{\boldsymbol{\alpha}}$ can correspond to different LME $\boldsymbol{\beta}_u$ and $\boldsymbol{\beta}_v$. This indicates $\mathbf{e}^T\mathbf{G}_u \neq \mathbf{e}^T\mathbf{G}_v$ while $\mathbf{c}^{T}\mathbf{G}_u=\mathbf{c}^{T}\mathbf{G}_v= \Tilde{\boldsymbol{\alpha}}$, which is a contradiction to Assumption \ref{Assumption_1}. Thus, Theorem \ref{LMP_LME} can be proved.
\end{proof}
By applying Theorem \ref{LMP_LME}, for the SCED problem \eqref{DCOPF}, we can directly use the LMP to obtain the LME reliably.


\section{Case Study}
In this section, we evaluate the performance of our proposed CRP-based LME identification across various testbeds. We show that our policy achieves superior computational efficiency compared to \emph{implicit function} and \emph{finite difference} methods, while also demonstrating robustness to load perturbations. Furthermore, we illustrate the behavior of LMP and LME in different critical regions and highlight the computational advantage of LMP-LME mapping over Load-LME mapping.
\begin{figure*}[htb]
	\centering
\includegraphics[width=0.99\linewidth]{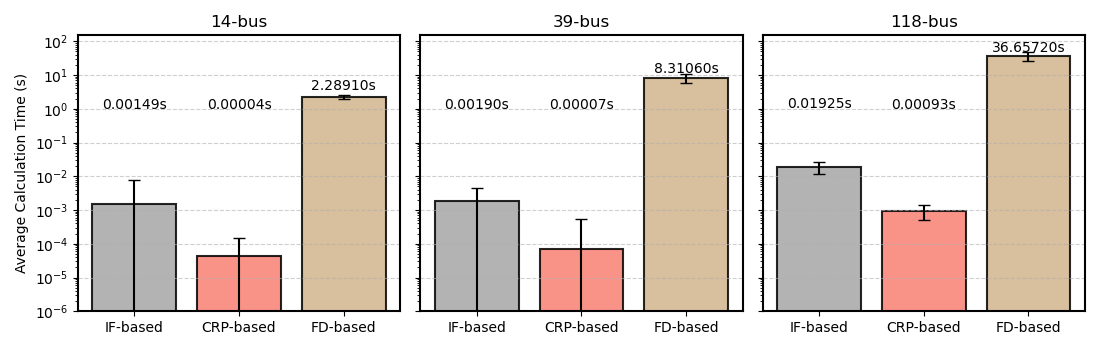}

\caption{Computation efficiency results. Reported time is averaged over 1000 samples with standard deviation shown as error bars. The y-axis is of logscale.} \label{Comparison_Time}
\end{figure*}
\begin{figure}[tb]
	\centering
\includegraphics[width=0.7\linewidth]{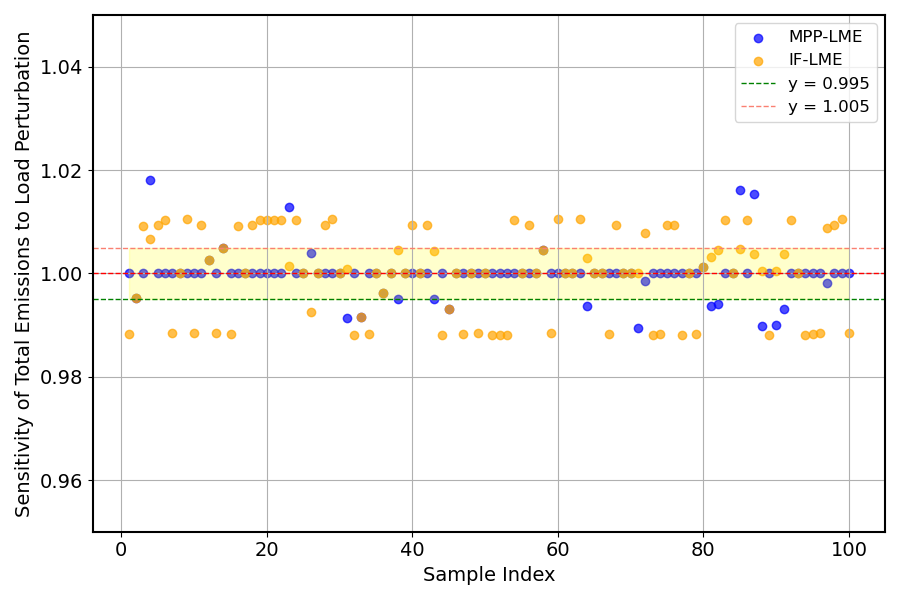}

\caption{Robustness Analysis for 118-bus system.} \label{Sensitivity analysis}

\end{figure}
\subsection{Benchmark Methods}
The methods used for our comparison are summarized as follows: 
\subsubsection{Implicit Function (IF)-based LME} This method is utilized in \cite{valenzuela2023dynamic}, which applies Lemma \ref{lemma_1} for each load sample to obtain $\frac{d\mathbf{x}(\boldsymbol{l})}{d\boldsymbol{l}}$ to calculate the LME via \eqref{LME_ori}. Note that, in this case, no critical region will be developed.
\subsubsection{Finite Difference (FD)-based LME} This method uses the finite difference as an approximation for the LME of each node as follows,
\begin{align}
\beta_j = \frac{E(\boldsymbol{l}+\Delta \boldsymbol{l}_j)-E(\boldsymbol{l})}{\Delta \boldsymbol{l}_j}. 
\end{align}
where $\Delta \boldsymbol{l}_j$ represents a sufficiently small increment in the load at node $j$. Note that to find $E(\boldsymbol{l}+\Delta \boldsymbol{l}_j)$ and $E(\boldsymbol{l})$, it is required to solve the SCED problem twice with the given load vectors to calculate the emissions for each load vector.
\begin{table}[tbp]
\centering
\caption{System configurations and the number of CR} \label{MPP_LME: Configurations}
\setlength{\tabcolsep}{0.6mm}{
\begin{tabular}{ccccc}
\hline
Systems & Generator (Type-Number)                               & Load Range ($\omega$)                 & No.CR \\ \hline
14      & Coal-1, NG-2, Wind-1, Solar-1     & \multirow{3}{*}{$\pm 30\%$} & \textbf{4}     \\
39      & Coal-2, NG-4, Wind-2, Solar-2     &                             & \textbf{15}    \\
118     & Coal-11, NG-21, Wind-11, Solar-11     &                             & \textbf{18}    \\ \hline
\end{tabular}
}
\end{table}
\subsection{Simulation Setup}
To verify the effectiveness of the proposed CRP-based LME policy, we perform case studies on the IEEE 14-, 39-, and 118-bus power systems. We simulate the cases where coal, natural gas, wind, and solar generators are used to power the grid and the generators’ carbon emission rates are 1000, 469, 12, and 46 $kg CO_2/MW$, respectively. We conduct the MPP procedure for the SCED problems of all systems and get the corresponding critical regions along with their gradients with respect to the load parameters. The system configurations and the number of critical regions are given in Table \ref{MPP_LME: Configurations}. The load profiles are based on the demand data from 2021/07/01 provided by CAISO \cite{CASIO}.

All simulations have been carried out on an unloaded MacBook Air with Apple M1 and 8G RAM. Specifically, the MPP-related and FD-related optimization problems are modeled and solved using the YALMIP toolbox and MPT3 toolbox in MATLAB R2022b. The IF-related optimization problems are modeled and solved using Julia 1.8.2.

\subsection{Computation Efficiency Analysis}
We use the CRP-based, IF-based, and FD-based methods to obtain the LMEs for unseen load samples and compare their average calculation time as shown in Fig. \ref{Comparison_Time}. The FD-based method takes the longest time in all cases, always thousands of times longer than the other methods, reaching up to 36.7 seconds per sample in the 118-bus system case. This is because it requires solving two SCED problems for each node to compute the finite difference. Compared to the IF-based method, our proposed method achieves significant acceleration, with speedups of 25.2x, 29.5x, and 22.0x for the 14-, 39-, and 118-bus systems, respectively. This demonstrates that mapping the load sample to the specific critical region is much faster than using the implicit function. In most cases, the three methods produce the same LME results. However, in the 39-bus system, the proposed method may yield different results for few samples due to slight deviations in the calculation of generation sensitivity, which are amplified by the emission rates. Following the resulting generation policy ensures valid LMEs, with only a slight difference in total emissions from the optimal value, as the generation policy is nearly identical.
\begin{figure*}[htbp]

	\centering
\includegraphics[width=1\linewidth]{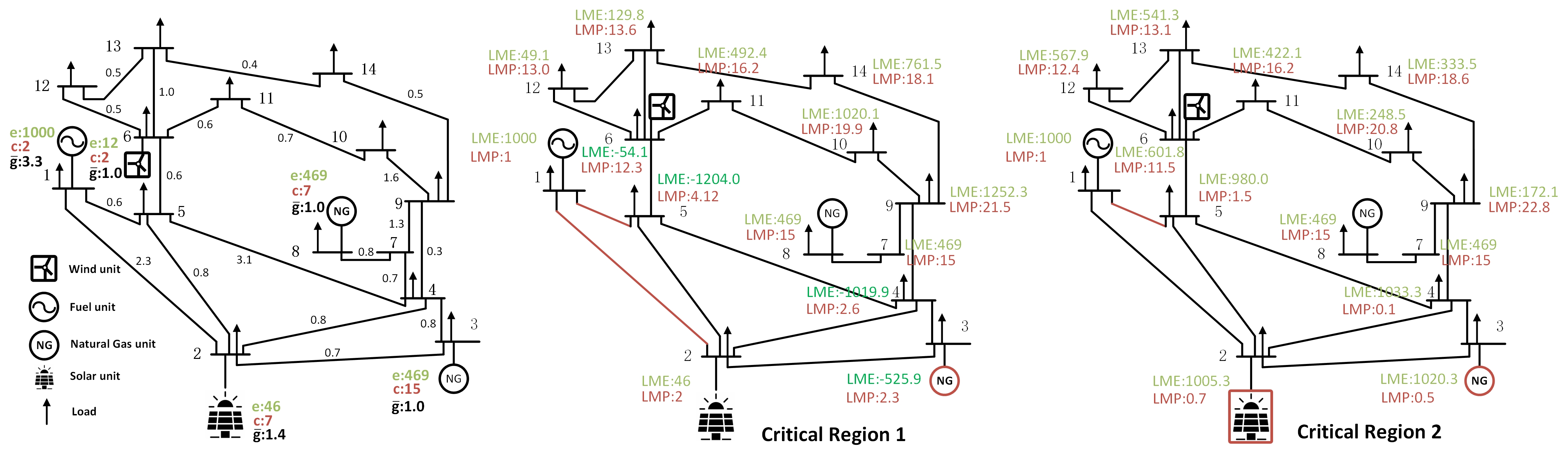}

\caption{LME and LMP in different critical regions of the 14-bus system. Negative LME are observed in Critical Region 1.} \label{LME_LMP_14}
\vspace{-1em}
\end{figure*}

\begin{figure}[tbp]
	\centering
\includegraphics[width=0.75\linewidth]{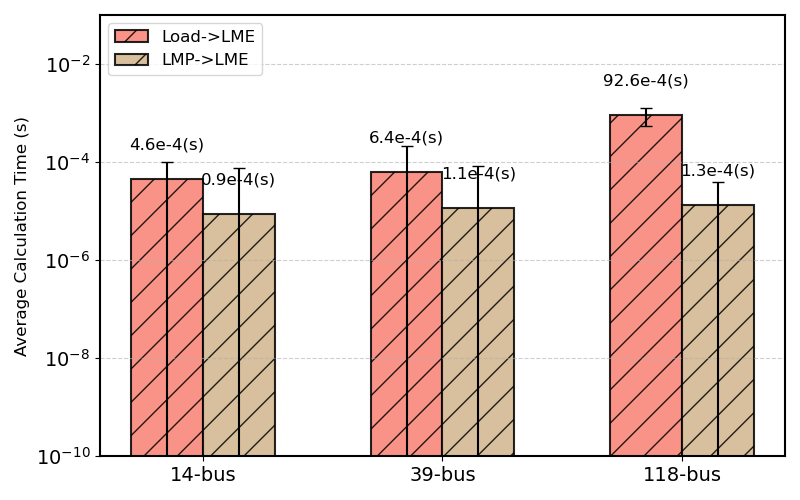}
\vspace{-0.5em}
\caption{Calculation time of Load-based and LMP-based LME mappings.} \label{Comparison_Time_LMP}
\vspace{-1em}
\end{figure}
\subsection{Robustness Analysis on Load Perturbations}
To assess the sensitivity of total emissions to load perturbations, we assume that, regarding the IF method, load samples within the $\pm 10\%$ operating range share the same LME. In contrast, our method allows for quick updates to the LME based on the corresponding critical regions. Then, we apply a $1\%$ perturbation to the load profile and calculate the total carbon emissions using the LME from both the IF-based and our methods. These results are then compared with the actual emissions obtained by solving the SCED models. As shown in Fig. \ref{Sensitivity analysis}, our method and the IF-based method enable 86\% and 42\% of the samples, respectively, to estimate total emissions with an error of less than 0.5\%. This suggests that although our method has the maximal error, likely due to the perturbation causing a shift in the critical region, it remains more robust to load perturbations in most cases. 

\subsection{LME and LMP in Each Critical Region}
Besides the LMEs, we also calculate and record the LMPs for each critical region to realize the mapping from LMP to LME. Here we show two interesting scenarios of the critical regions for the 14-bus system in Fig. \ref{LME_LMP_14}. The LMEs of nodes 3, 4, 5, and 6 are negative since the limits of line 1-2 and line 1-5 are active while the solar unit and wind unit are still surplus. This means that when the load at these nodes increases within critical region 1, congestion may occur on line 1-5 or line 1-2. As a result, other nodes will rely more on generators with lower emission rates based on SCED results, thereby reducing the total carbon emissions and leading to negative LME. In contrast, within the critical region 2 where line 1-2 will not be congested, more generation from coal units tends to be utilized due to their lower cost, resulting in the increase of LMEs in nodes 3,4,5 and 6. Thus, in practice, we can adjust the load profile to change its critical region and LME, and further shift the load to nodes with negative LMEs to reduce total carbon emissions.

\subsection{Comparison Between Load-based and LMP-based LME}
As shown in Fig. \ref{LME_LMP_14}, the LMPs vary across different critical regions. Thus, we can directly use the LMPs of load samples to locate their corresponding critical regions and determine the LMEs. The results in Fig. \ref{Comparison_Time_LMP} demonstrate that LMEs can be obtained within 0.01s for all cases, with the LMP-LME mapping being faster than the Load-LME mapping. Specifically, the acceleration factors are 5.1x, 5.8x, and 71.8x for the 14-bus, 39-bus, and 118-bus systems, respectively. This indicates that locating critical regions by identifying the load vector’s position within the polyhedron is more complex than simply matching the sample’s LMP to its LME.

\section{Conclusion and Future Works}
In this paper, we propose a novel LME policy for power system operations based on critical region projection (CRP). By solving multi-parametric programming for underlying SCED problems, we can derive specific critical load regions corresponding to the load operating range. Loads within each region share the same generation sensitivity. This sensitivity links the LME to the load region, as well as the LMP to the load region, enabling direct mapping from LMP to LME. In all testing cases with various system configurations, our CRP-based LME calculation is significantly faster than other benchmarked methods, with the LMP-to-LME mapping further accelerating the process compared to load-to-LME mapping. More interestingly, such LME derivation procedures reflect how LME is impacted by line congestion and generation dispatch in an explicity way. In the future work, we will extend this approach to more complex problems, such as SCED models with quadratic costs and dynamic constraints. We are also interested in validating LME calculations by incorporating real-world grid and emission data, which can give implications for emission reduction and sustainable operations of power networks.
\newpage
\bibliographystyle{IEEEtran}
\bibliography{bib}

\end{document}